\documentclass[10pt,journal,column,romanappendices]{article}

\usepackage{amsthm}
\usepackage{amssymb} 
\usepackage{amsmath}
\usepackage{graphicx}
\usepackage{stmaryrd}

\newtheorem{theorem}{Theorem}
\newtheorem{definition}{Definition}
\newtheorem{example}{Example}

\newtheorem{lemma}[theorem]{Lemma}
\newtheorem{corollary}[theorem]{Corollary}

\theoremstyle{remark}
\newtheorem*{remark}{Remark}

\newcommand{\defeq}{\overset{\text{def}}{=}}

\newcommand{\qedwhite}{\hfill \ensuremath{\Box}}

\begin{document}

\title{Bits Through Queues With Feedback}

\author{Laure Aptel and Aslan Tchamkerten \\
Telecom ParisTech}

\maketitle

\begin{abstract}
In their $1996$ paper Anantharam and Verd\'u showed that feedback does not increase the capacity of a queue when the service time is exponentially distributed. Whether this conclusion holds for general service times has remained an open question which this paper addresses.
  
Two main results are established for both the discrete-time and the continuous-time models. First, a sufficient condition on the service distribution for feedback to increase capacity under FIFO service policy. Underlying this condition is a notion of weak feedback wherein instead of  the queue departure times the transmitter is informed about the instants when packets start to be served.  
Second, a condition in terms of output entropy rate under which feedback does not increase capacity. This  condition is general in that it depends on the output entropy rate of the queue but explicitly depends neither on the queue policy nor on the service time distribution. This condition is satisfied, for instance, by queues with LCFS service policies and bounded service times.

\end{abstract}

\section{Introduction}
In \cite{bits_through_queues} Anantharam and Verd\'{u} investigated a class of timing channels modeled as single-server queues. Under a First In First Out (FIFO) service policy they showed that, for a fixed queue output rate $\lambda$, the capacity $C(\lambda)$ and the feedback capacity $C_F(\lambda)$  satisfy 
$$C(\lambda) \leq \sup_{W \geq 0 \atop \mathbb{E}[W] \leq \frac{1}{\lambda} - \frac{1}{\mu}} I(W ; W+S) = C_F(\lambda)\qquad \lambda<\mu$$
where  $S$ denotes the random service time of given rate $\mu$. 
They also showed that when the service distribution is exponential the above inequality becomes an equality and therefore revealing the transmitter the queue departure times through feedback does not help. Whether this negative result holds for other combinations of service policy and service distribution ever since has remained an open issue. 

In this paper we investigate single-server queues for both the discrete-time and the continuous-time models and provide sufficient conditions for feedback to increase capacity and for feedback not to increase capacity.\footnote{Throughout the paper we only consider single-server queues.} A central piece of our investigation is  a notion of weak feedback which describes the situation where the transmitter is causally revealed the instants when packets start getting served---as opposed to the departure times under regular feedback. 

\subsection*{Main contributions}
Let $C_{WF}$ and $C_{WF}(\lambda)$ denote the capacity and the capacity at fixed output rate $\lambda$, respectively, when weak feedback is available---this notion will be made precise in Section~\ref{inccap}.
\begin{itemize}
\item Theorem~\ref{3b}:  Weak feedback represents an intermediate stage between feedback and no feedback, that is
 $$C(\lambda) \leq C_{WF}(\lambda) \leq C_F(\lambda)\qquad \lambda<\mu.$$
 \item Theorem~\ref{wfc}:  An upper bound to the weak feedback capacity is
$$C(\lambda)\leq \lambda\sup_{X : \mathbb{E}[W(X)] \atop W(X)= (X-S_1)^+} I(W(X) ; W(X) + S_2)\qquad \lambda <\mu$$
where $(X,S_1,S_2)$ are jointly independent and where $S_1$ and $S_2$ follow the service time distribution.
\item Theorem~\ref{suffcond}: If the inequality $$\sup_{X:\mathbb{E}[W(X)]\leq\frac{1}{\lambda}-\frac{1}{\mu}\atop W(X)=  (X-S_1)^+}H((X-S_1)^++S_2) < \sup_{W \geq 0 \atop \mathbb{E}[W]\leq\frac{1}{\lambda}-\frac{1}{\mu}} H(W+S)$$
is satisfied for any $\lambda<\mu$ (here $(X,S_1,S_2)$ are as above and $W$ is independent of $S$), then $$C \leq C_{WF} < C_F$$
and therefore feedback increases capacity. Examples of services times that satisfy the above inequality include  a Bernoulli($1/2$) service time in the discrete-time model and a uniformly distributed service time in the continuous-time model.
\item Corollary~\ref{th:Gal}: 
Under Last Come First Served policy and bounded service times
$$C(\lambda)=C_{WF}(\lambda)=C_F(\lambda)\qquad \lambda<\mu$$ and therefore feedback does not help.
\end{itemize}

\subsection*{Related works}
In \cite{discrete_bits_through_queues} the authors investigated the discrete-time version of Anantharam and Verd\'u's model and also investigated the variant where the server can simultaneously serve multiple packets. In \cite{timing_capacity_of_discrete_queues} the entropy increasing property of queueing systems was investigated. In particular, this work established sufficient conditions on queues under which the output entropy is greater than the input entropy. One of these conditions is key to derive Corollary~\ref{th:Gal} above. 

The single-server queue with exponential service time is arguably the most well-studied timing channel. Beyond capacity, robustness with respect to service time noise, sequential decoding, and cutoff rate were investigated  in \cite{robust_decode} and \cite{algo_decode}. The reliability function was investigated in \cite{algo_code} and \cite{WagAn}.

Several variations of the single-server queue channel have been considered. For instance, in \cite{algo_multiserver} the authors investigated multi-server queues and in \cite{bufferless} the authors investigated a queue without buffer wherein arriving packets are dropped if the queue is already serving a packet. 

Secure and covert communication over timing channels was explored in a number of papers, {\it{e.g.}}, \cite{secure_btq}, \cite{covert_btq}, \cite{game_theory_covert}, and \cite{biswas2017survey} for a comprehensive survey. A recent application of timing channels in the context of energy harvesting systems can be found in  \cite{harvesting_channel}. 

The paper is organized as follows. In Section~\ref{back} we briefly review the model in \cite{bits_through_queues}. In Section~\ref{risul} we present our results and prove them in Section~\ref{proofs}. Section~\ref{concrem} concludes the paper.

\section{Background}\label{back}
\begin{figure}
\centering
\begin{picture}(280,60)
	\put(60,0){\includegraphics[scale=1]{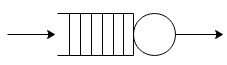}}

	\put(120,46){Queue}

	\put(77,30){$a_i$}
	\put(205,30){$d_i$}
\end{picture}
	\label{queue}
	\caption{Timing channel modeled as a single-server queue.}
\end{figure}
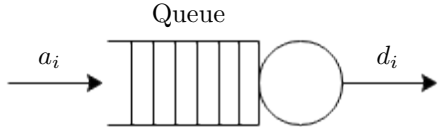

The channel, depicted in Fig.~\ref{queue}, models timing information transmission through a queue.  The $i$\/th packet arrives at the queue at time $a_i$ and departs the queue at time $d_i$. The channel captures timing uncertainty by means of a single-server queue with random service time $S_i$, {\it{i.e.}} $$d_i = a_i + \Delta_i +S_i$$ where $\Delta_i$ denotes the time spent by the $i$\/th packet waiting in the queue (if the queue is empty when the $i$\/th packet arrives then it is immediately served and $\Delta_i = 0$ ). The $S_i$'s are assumed to be independant and identically distributed with mean $1/\mu$, where $\mu$ denotes the service rate of the queue.  The $\Delta_i$'s depend on the service discipline (the choice of the next packet to be served in the queue). Throughout this paper we suppose that $\mu>0$ and unless stated otherwise we assume FIFO service policy.

When feedback is available (see~Fig.~\ref{picture with feedback}) the sender is causally revealed the departure times and can thus choose $a_i$ as a function of both the message to be transmitted and past departure times. 
\begin{definition}[Code]
An $(n,M,T, \epsilon)$-code for a timing channel used without feedback consists of  
\begin{itemize}
\item $M$ codewords such that each codeword is a vector of $n$ arrival instants $a^n$ (with $a_1 \leq a_2 \leq \cdots \leq a_n$) and such that the $n$\/th departure from the queue occurs on average (over equiprobable codewords and the queue distribution) no later than $T$;  
\item a decoder that upon observing $n$ departures $d^n$ produces an estimate $\hat{u}^n=\hat{u}(d^n)$ of $u$ with an average error probability (over equiprobable codewords and the queue distribution) satisfying
 $$\mathbb{P}(\hat{u} \neq u) \leq \epsilon. $$
\end{itemize}
The rate of an $(n, M, T, \epsilon)$-code is defined as\footnote{Throughout the paper logarithms are to the base two.} $$\frac{\log M}{T} \quad \text{bits/second}.$$

If feedback is available then $a_i$ may, in addition to the message to be transmitted, also causally depend on the departure times in the sense that, given $d^n$ event $\{a_i=t\}$ may depend on $\{d_k: d_k<t\}$ but must be independent of $\{d_k: d_k\geq t\}$. 
\end{definition}

\begin{figure}
\centering
\begin{picture}(400,80)
	\put(17,27){\includegraphics[scale=0.87]{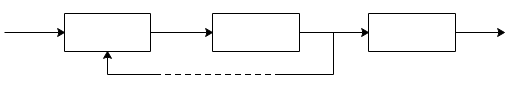}}
	\put(66,59){\large{Encoder}}
	\put(168,59){\large{Queue}}
	\put(267,59){\large{Decoder}}

	\put(-5,69){$u\in\{1,...,M\}$}
	\put(116,69){$a^n(u,d^n)$}
	\put(232,69){$d^n$}
	\put(328,69){$\hat{u}$}
\put(140,38){causally}
\end{picture}
\vspace{-1.2cm}
	\caption{A queue with feedback. A message $u$ is encoded into arrival times $a^n(u)$ with the additionnal information of the causally received $d_k$'s. Given departure times $d^n$ the decoder produces a message estimate $\hat{u}$. }
	\label{picture with feedback}
\end{figure}
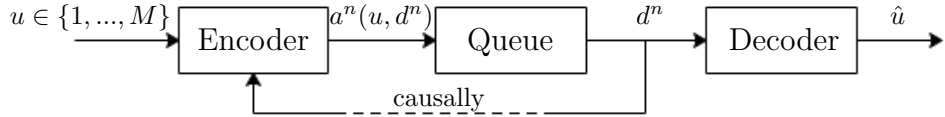

\begin{definition}[Capacity]
The capacity $C$(resp. $C_F$) of a timing channel used without feedbback (resp. with feedback) is the highest R for which for all $\gamma > 0$ there exists a sequence of $(n, M, T, \epsilon_T)$-codes (resp. feedback codes) such that 
$$\frac{\log M}{T} > R -  \gamma$$
and $\epsilon_T \to 0$. 
\end{definition}
\begin{definition}[Capacity at fixed output rate]
For any $\lambda$, the capacity $C(\lambda)$ (resp. $C_F(\lambda$)) of a timing channel used without feedback (resp. with feedback) at output rate $\lambda$ is the highest R for which for all $\gamma > 0$ there exists a sequence of (n, M, $n/\lambda$, $\epsilon$)-codes (resp. feedback codes) such that 
$$\lambda \frac{\log M}{n} > R -  \gamma$$
and $\epsilon \to 0$. 
\end{definition}
We clearly have 
$$C\leq C_F$$
and $$C(\lambda) \leq C_F(\lambda)\quad \lambda >0.$$

\begin{theorem}[\cite{bits_through_queues}]
We have
 $$C(\lambda) \leq C_F(\lambda)\quad \lambda <\mu$$
 and
$$C = \sup_{\lambda < \mu}C(\lambda) \leq  C_F = \sup_{\lambda < \mu}C_F(\lambda).$$
\end{theorem}
When feedback is available capacity admits an explicit expression:  
\begin{theorem}[see (3.4) in  \cite{bits_through_queues}]\label{btqf}
We have
 $$ C_F(\lambda) = \lambda \sup_{W \geq 0   \atop   \mathbb{E}[W] \leq \frac{1}{\lambda} - \frac{1}{\mu}} I(W ; W + S)\qquad \lambda<\mu.$$
\end{theorem}
Note that the above two theorems hold whether time is continuous or discrete. Throughout this paper results hold in both settings unless stated otherwise.

\section{Results}\label{risul}
This section is split into two parts. Section~\ref{inccap} focuses on a sufficient condition under which feedback increases capacity and Section~\ref{noninccap} focuses on a sufficient condition under which feedback does not increase capacity.
\subsection{Weak feedback} \label{inccap}
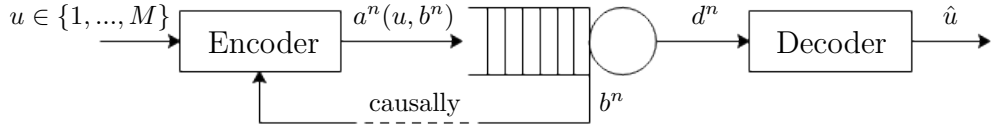
\begin{figure*}[t!]
\centering
\begin{picture}(400,100)
	\put(20,20){\includegraphics[scale=0.8]{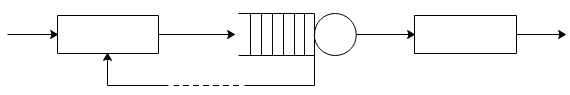}}

	\put(65,52){\large{Encoder}}
	\put(280,52){\large{Decoder}}

	\put(120,63){$a^n(u, b^n)$}
	\put(213,30){$b^n$}
	\put(248,63){$d^n$}

	\put(-10,63){$u\in\{1,...,M\}$}
	\put(343,63){$\hat{u}$}

	\put(126,30){causally}
\end{picture}
	\caption{A queue with weak feedback. A message $u$ is encoded into arrival times $a^n(u)$ with the additionnal information of the causally received $b_k$'s. Given departure times $d^n$ the decoder produces a message estimate $\hat{u}$.}
	\label{picture with weak feedback}
\end{figure*}

When feedback is available the transmitter causally knows the departure times $$d_i=a_i + \Delta_i +S_i.$$ We now introduce a slightly different notion of feedback. We say that weak feedback is available if the transmitter has causally access to $$b_i \defeq a_i + \Delta_i = d_i - S_i,$$ the time when the $i$\/th packet starts to be served. Note that since we assume a FIFO service discipline we have
\begin{align}\label{bprop}
b_i=\max\{d_{i-1},a_i\}.
\end{align} Indeed, when $a_i\geq d_{i-1}$ the $i$\/th packet is served immediately and therefore $b_i=a_i$. Instead, if $a_i < d_{i-1}$ then $b_i=d_{i-1}$ since the $i$\/th packet will have to wait until the previous packet exits the queue. Therefore, the feedback capacity is at least as large as the weak feedback capacity. Denoting by $C_{WF}$ and $C_{WF}(\lambda)$ the weak feedback capacity and the weak feedback capacity at fixed output rate $\lambda$, respectively, the following theorem follows: 

\begin{theorem}\label{3b} 
We have
$$C(\lambda) \leq C_{WF}(\lambda) \leq C_F(\lambda)\qquad \lambda<\mu.$$
and hence
$$C \leq C_{WF} \leq C_F$$

\end{theorem}

In \cite{bits_through_queues} it is shown that when feedback is available  encoding strategies may, without loss of optimality, be restricted to those where $$a_i \geq d_{i-1}\qquad i \in \{ 1, \cdots, n\},$$ that is to those where the queue is always empty since a packet arrival never occurs before the previous packet exits the queue. Specifically, for any strategy that allows the queuing of packets the distribution of the vector of departures given the sent message ${\mathbb{P}}_{{d^n}|U}$ is the same as that of a modified strategy where the $i$\/th arrival is not allowed to happen before the $(i-1)$\/th packet starts being served. Hence the joint distribution of the input and the output of the channel ($\mathbb{P}_{{d^n}|U}\mathbb{P}_U$) remains the same and so does the error probability---for any given decoding rule. 

Under weak feedback encoding strategies may be restricted to those that allow the queuing of at most one packet:
\begin{theorem} \label{bpos}
Without loss of optimality, under weak feedback encoding strategies may be restricted to those that satisfy $a_i \geq b_{i-1}$, $i \in \{ 2, \cdots, n\}$. 
\end{theorem}

The following theorem provides a multi-letter upper bound on the weak feedback capacity:
\begin{theorem}\label{wfcapacity}We have
$$C_{WF}(\lambda) \leq \lambda \left(  \left( \lim_{n \to \infty} \sup_{X^n \geq 0 \atop \frac1n \sum_{i=1}^n\mathbb{E}[D_i] \leq \frac{1}{\lambda} } \frac1n H(D^n) \right) - H(S) \right) \qquad \lambda<\mu$$
where $$ D_i \defeq  (X_i-S_{i-1})^+ + S_i$$ and where $X^n$ is independent of $S^n$ (i.i.d. service times).
\end{theorem}

The next theorem provides a single-letter upper bound to the weak feedback capacity:
\begin{theorem}\label{wfc}
We have
\begin{align*}
 C_{WF}(\lambda) \leq \lambda \sup_{X:\mathbb{E}(W(X))\leq1/\lambda-1/\mu  \atop W(X) \defeq (X-S_1)^+ } I(W(X) ; W(X)+S_2)\qquad \lambda<\mu
\end{align*}
where $S_1$ and $S_2$ are two independent random variables following the service time distribution and where $X$ is an arbitrary real random variable that is independent of $(S_1,S_2)$.
\end{theorem}
In the feedback capacity expression of Theorem~\ref{btqf} the $W$ can be interpreted as the waiting time of the queue. And because of feedback, the transmitter can control $W$ perfectly. Instead, if weak feedback is available $W$ can no longer be perfectly controlled but only up to a service time. This intuition is reflected by the fact that the transmitter now controls $W$ through $X$ via $W=(X-S_1)^+$.

\begin{corollary}
\label{bound}
We have
\begin{align*}
 C(\lambda) \leq \lambda \sup_{X:\mathbb{E}(W(X))\leq1/\lambda-1/\mu  \atop W(X) =(X-S_1)^+ } I(W(X) ; W(X)+S_2)\qquad \lambda<\mu
\end{align*}
where $S_1$ and $S_2$ are two independent random variables following the service time distribution and where $X$ is an arbitrary real random variable that is independent of $(S_1,S_2)$.
\end{corollary}
\begin{proof}[Proof of Corollary~\ref{bound}]
The bound holds since $ C(\lambda) \leq  C_{WF}(\lambda)$. A direct proof of the Corollary that bypasses the inequality $ C(\lambda) \leq  C_{WF}(\lambda)$ is given in the appendix.
\end{proof}

Notice that since for $\lambda<\mu$
\begin{align*}
C_{WF}(\lambda)\leq \lambda \sup_{X:\mathbb{E}(W(X))\leq1/\lambda-1/\mu \atop W(X) \defeq (X-S_1)^+} &I(W(X);W(X)+S_2) \\ & \hspace{-1cm}\leq \lambda \sup_{W \geq 0   \atop   \mathbb{E}[W] \leq \frac{1}{\lambda} - \frac{1}{\mu}} I(W ; W + S)\\
  &\hspace{-1cm} \overset{\text{Th.$2$}}{=} C_F(\lambda),
\end{align*}
Corollary~\ref{bound} provides an upper bound to the capacity without feedback which is at least as good as the feedback capacity bound.

The next theorem provides a (single-letter) sufficient condition under which feedback increases capacity:
\begin{theorem}\label{suffcond}
Let $S$ be a rate $\mu$ service time that satisfies
\begin{multline*}
\sup_{X:\mathbb{E}[W(X)]\leq\frac{1}{\lambda}-\frac{1}{\mu}\atop W(X)=  (X-S_1)^+}H((X-S_1)^++S_2) < \sup_{W \geq 0 \atop \mathbb{E}[W]\leq\frac{1}{\lambda}-\frac{1}{\mu}} H(W+S)
\end{multline*}
where $\lambda>\mu$, where $S_1$ and $S_2$ are independent random variables that follow the service time distribution,  where $X$ is an arbitrary real random variable that is independent of $(S_1,S_2)$, and where
$W$ is independent of $S$.
Then $$C(\lambda) < C_F(\lambda).$$
\end{theorem}
Note that Theorem~\ref{suffcond} involves two convex optimization problems which can easily be solved with numerical methods. 
\begin{example}\label{binary}
The hypothesis of Theorem~\ref{suffcond} are satisfied in the discrete-time model if $\mathbb{P}(S=1) =\mathbb{P}(S=2) = 1/2$ and in the continuous-time model if $S$ is uniform over $[0,1]$. The proof is given in Section~\ref{proofs}.\end{example}

\subsection{When feedback does not increase capacity} \label{noninccap}

The following result provides a general condition under which feedback does not increase capacity. Let 
$$A_i \defeq a_i - a_{i-1}$$ and $$D_i \defeq  d_i - d_{i-1}$$
denote the interarrival and the interdeparture times,
respectively.

\begin{theorem}\label{th:suffcondequal}
In the discrete-time model, for a given rate $\mu>0$  suppose  that for any $ \lambda< \mu$$$H(g_{\lambda}) \leq \liminf_{n \to \infty} \frac{H(D^n)}{n}$$ where 
$$H(g_{\lambda})\defeq\frac{-\lambda \log \lambda -(1-\lambda)\log (1-\lambda)}{\lambda}$$
denotes the entropy of the geometric distribution $g_{\lambda}$ with mean ${1}/{\lambda}$ and where $D^n$ denotes the interdepartures that correspond to a sequence of interarrivals $A^n$ that are i.i.d. according to a mean $1/\lambda$ geometric distribution. Then 
$$C_F(\lambda) = C(\lambda)=\lambda (H(g_\lambda) - H(S)) \quad \lambda \leq \mu.$$
In the continuous-time model the result holds verbatim by replacing the geometric distribution with the exponential distribution and its corresponding entropy.
\end{theorem}

Under LCFS (Last Come First Served) service discipline the queue always selects the last arrived packet to be served. If a packet arrives and the queue is busy the current service is interrupted and the queue starts serving the newly arrived packet (see~\cite{bits_through_queues}). 
\begin{corollary} \label{th:Gal}
For a single-server queue with LCFS discipline and bounded service time $C_F=C$.
\end{corollary}
\begin{proof}[Proof of Corollary~\ref{th:Gal}]
As a Corollary of a result by Gallager and Prabhakar (see \cite[Theorem~$2$]{timing_capacity_of_discrete_queues}) we have that LCFS  service discipline together with a bounded service time satisfies the hypothesis of Theorem~\ref{th:suffcondequal}.
\end{proof}

\begin{remark}
In \cite{bits_through_queues} Anantharam and Verd\'u also investigated the case where information can simultaneously be conveyed via bit timing and bit value. The situation is depicted in Fig.\ref{information-bearing} where $x^n$ and $y^n$ represent the input and the output vector of a discrete memoryless channel. Under FIFO service policy they showed that the capacity of the queue with information bearing symbols is
$$C_I = \sup_{\lambda<\mu}[C(\lambda)+\lambda C_0]$$
where $C_0$ denotes the capacity of the discrete memoryless channel. In fact, the above result holds also when feedback is available by replacing $C(\lambda)$ with $C_F(\lambda)$ (see \cite[Section IV Theorem 9]{bits_through_queues}). Following the same arguments one can easily conclude that the result also holds under weak feedback by replacing $C(\lambda)$ with $C_{WF}(\lambda)$.
\begin{figure}
	\centering
	\includegraphics[scale=0.65]{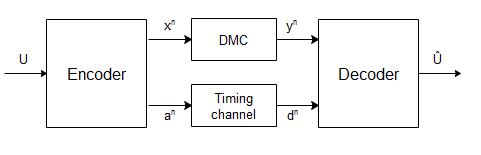}
	\caption{Information-bearing queue.}
	\label{information-bearing}
\end{figure}

\end{remark}

\section{Proofs}
In this section we limit ourselves to the discrete-time model. The proofs in the continuous time setting amount to changing probability mass functions to densities and replacing the geometric distribution with the exponential distribution---since the exponential distribution, similarly as the geometric distribution in the discrete-time setup, maximizes entropy among all continuous random variables with the same mean.
\label{proofs}

\subsection{Proof of Theorem \ref{bpos}}

To prove the theorem we show that, for any strategy that allows the queuing of more than one packet the distribution of the vector of departures given the sent message $\mathbb{P}(d^n|u)$ is the same as that of a modified strategy where the $i$\/th arrival is not allowed to happen before the $(i-1)$\/th starts being served. Hence the joint distribution between the input and the output of the channel ($\mathbb{P}(u,d^n)$) remains the same and so does the error probability (for any decoding rule).

Fix an encoding strategy. Given a realisation $a^n$ of arrivals we denote by $\tilde{a}^n$ the modified realisation where
\begin{align*}
\tilde{a}_i= \max\{a_i, b_{i-1}\} .
\end{align*}
We have 
\begin{align}\label{pn}
\mathbb{P}(d^n|u) &= \prod_{i=1}^n \mathbb{P}(d_i|u,d^{i-1}).
\end{align}
Observe that $a_1$ is uniquely determined by $u$ and that  $b_1=a_1$ as the queue is supposed to be initially empty. Then for $i \geq 2$ arrival $a_i$ depends on both  $u$ and the weak feedback $b^{i-1}$ where
\begin{align*}
b_i=\max(a_i, d_{i-1}).
\end{align*}
Thus, from $u$ and $d^{i-1}$ we can recursively compute $(a_1,b_1),(a_2,b_2),\ldots ,(a_i,b_i)$ and
it follows that
\begin{align}
\mathbb{P}(d_i|u,d^{i-1})  &=\mathbb{P}(d_i|u, d^{i-1}, b^i, a^{i}) \nonumber \\
  &=\mathbb{P}(S_i = d_i - \max(a_i,d_{i-1} ) | u, d^{i-1}, b^i, a^{i})\nonumber \\
  &=\mathbb{P}(S_i = d_i - \max(a_i,d_{i-1} )) \label{pdi}
\end{align}
where the second equality holds because
\begin{align*}
d_i &= b_i + S_i \\
  &= \max(a_i, d_{i-1}) +S_i
\end{align*}
by \eqref{bprop}.
If we change $a_i$ to $\tilde{a}_i$ then either $$a_i \geq b_{i-1}$$ and $$a_i = \tilde{a}_i$$ or $$a_i<b_{i-1} = \tilde{a}_i \leq d_{i-1}.$$ Hence in both cases $$\max(a_i,d_{i-1} )  = \max(\tilde{a}_i,d_{i-1} ). $$
It follows that if we change $a_i$ to $\tilde{a}_i$ then the conditional probability $\mathbb{P}(d_i|u,d^{i-1})$ remains the same by \eqref{pdi} and so will the distribution of $d^n$ conditioned on $u$ by \eqref{pn}. \qedwhite

\subsection{Proof of Theorem \ref{wfcapacity}}
To prove the theorem, let $U\in \{1, \cdots , M\}$ denote the uniformly chosen message to be transmitted and let  $V\in \{1, \cdots , M\}$ denote the decoded message. Then
\begin{align*}
I(U;V)  &\overset{(a)}{\leq} I(U; D_1,\cdots, D_n) \\
  &=\sum_{i=1}^nI(U;D_i|D^{i-1}) \\
  &\overset{(b)}{\leq} \sum_{i=1}^nI(W_i;D_i|D^{i-1}) \\
  &= \sum_{i=1}^n \left( H(D_i|D^{i-1}) - H(D_i|W_i,D^{i-1}) \right) \\
  &= \sum_{i=1}^n \left( H(D_i|D^{i-1}) - H(S_i) \right) \\
  &{=}\sum_{i=1}^n \left( H(D_i|D^{i-1}) \right) - n H(S) \\
  &= H(D^n) - n H(S) \\
  &\leq n \left( \left( \sup_{X^n \geq 0 \atop \frac1n \mathbb{E}[\sum_1^nD_i] \leq \frac{1}{\lambda} } \frac1n H(D^n) \right) - H(S) \right).
\end{align*}
where $W_i$ denotes the waiting time of the queue between the departure of the $(i-1)$\/th packet and the beginning of the service of the $i$\/th packet:
$$W_i\defeq D_i - S_i.$$
Inequality $(a)$ holds because of the Markov  chain $U\rightarrow D^n\rightarrow V$. Inequality $(b)$ holds since conditioned on $D^{i-1}$ we have the Markov chain $U \rightarrow W_i \rightarrow D_i.$ 
Indeed,  $D_i$ is only a function of $S_i$ when conditioned on $W_i$ and, in turn, $S_i$ is independent of $D^{i-1}$ and $U$.

From Fano's inequality and the data processing inequality \cite{Cover_Thomas} we have for any $(n,M,n/\lambda,\epsilon)$-code
$$\log M \leq \frac{1}{1-\epsilon}[I(U;V) +1]$$
and it follows that 
\begin{align*}
\lambda \frac{\log M}{n} \leq \frac{\lambda}{1-\epsilon}  	\left(    \left( \sup_{X^n \geq 0 \atop \frac1n\mathbb{E}[\sum_1^n D_i] \leq \frac{1}{\lambda} } \frac1n H(D^n) \right) - H(S)+ \frac{1}{n} \right).
\end{align*}
The theorem follows.

\subsection{Proof of Theorem \ref{wfc}}
From Theorem~\ref{wfcapacity} we have
$$C_{WF}(\lambda) \leq \lambda \left(  \lim_{n \to \infty} \sup_{X^n \geq 0 \atop \frac1n \mathbb{E}[\sum_1^n D_i] \leq \frac{1}{\lambda} } \left(  \frac1n H(D^n) - H(S) \right)  \right)$$
where $ D_i=  (X_i-S_{i-1})^+ + S_i$.
Now,\begin{align*}
\frac1n H( D^n) -H(S)  &\leq \frac1n\left( \sum_{i=1}^n H(D_i) \right) - H(S)\\
  &=  \frac1n  \sum_{i=1}^n \left[H((X_i-S_{i-1})^++S_i) - H(S) \right] \\
  &{=}  \frac1n \sum_{i=1}^n \left[ H((X_i-S_{i-1})^++S_i) - H(S_i) \right]\\
  &=\frac1n \sum_{i=1}^n \left[ H(W(X_i)+S_i) - H(W(X_i)+S_i| W(X_i) \right]\\
  &= \frac1n \sum_{i=1}^n  I(W(X_i) ; W(X_i)+S_i)
\end{align*}
where $W(X_i) \defeq (X_i-S_{i-1})^+$. Letting 
$$c(a) \defeq \sup_{X: \mathbb{E}[W(X)] \leq a \atop W(X)=  (X-S_1)^+} I(W(X) ; W(X)+S_2)$$
we have
\begin{align*}
I(W(X_i) ; W(X_i)+S_2)\leq c(\mathbb{E}[W(X_i)]) .
\end{align*}
Assuming $c$ to be concave an non-decreasing---see Lemma~\ref{concave} below---it follows that
\begin{align*}
\frac1n H(D^n) -H(S) &\leq \frac1n \sum_{i=1}^n c(\mathbb{E}[W(X_i)]) \\
  &\leq c\left(\frac1n \sum_{i=1}^n \mathbb{E}[W(X_i)]\right)\\
  &= c\left( \mathbb{E}\left[\frac1n\sum_{i=1}^n W(X_i)\right] \right)\\
  &=c\left( \mathbb{E}\left[\frac1n\sum_{i=1}^n (D_i-S_i)\right] \right)\\
  &=c\left( \mathbb{E}\left[\frac1n\sum_{i=1}^n D_i\right]-\mathbb{E}\left[\frac1n\sum_{i=1}^n S_i\right] \right)\\
  &\leq c \left( \frac{1}{\lambda} - \frac{1}{\mu} \right)
\end{align*}
since $\mathbb{E}[S_i] = 1/\mu$ and $\frac1n \mathbb{E}[\sum_1^n D_i] \leq \frac{1}{\lambda}$. The theorem follows. \qedwhite
\begin{lemma}\label{concave}
The function
$$c(a) = \sup_{X:\mathbb{E}[W(X)] \leq a \atop W(X)=  (X-S_1)^+} I(W(X) ; W(X)+S_2).$$
where $(X,S_1,S_2)$ are independent and where $S_1$ and $S_2$ are distributed according to the service time is concave and non-decreasing. 
\end{lemma}

\begin{proof}
It is clear that $c$ is non-decreasing. Now let $a < b$ be non-negative numbers, let $\alpha \in [0,1]$, and let $X_1$ and $X_2$ be two independent random variables that achieve $c(a)$ and $c(b)$, respectively. We assume here that the suprema are achieved. At the end of the proof we provide the slight extension to the case where the suprema are not achieved.  
 
Define integer valued random variable $X$ such that 
$$\mathbb{P}(X=k) = \alpha \, \mathbb{P}(X_1 = k) + (1-\alpha) \, \mathbb{P}(X_2 = k).$$
Then for  $n > 0$
\begin{align*}
\mathbb{P}(W(X)=n) &= \mathbb{P}((X-S)^+=n) \\
  &= \sum_{k=1}^{\infty} \mathbb{P}(S=k) \cdot \mathbb{P}(X=n+k) \\
  &= \sum_{k=1}^{\infty}  \mathbb{P}(S=k) \, (\alpha \cdot\mathbb{P}(X_1=n+k) + (1-\alpha) \cdot \mathbb{P}(X_2 = n+k)) \\
  &= \alpha \, \sum_{k=1}^{\infty} \mathbb{P}(S=k)\cdot \mathbb{P}(X_1=n+k)  + (1-\alpha) \, \sum_{k=1}^{\infty} \mathbb{P}(S=k)\cdot\mathbb{P}(X_1=n+k) \\
  &= \alpha \cdot \mathbb{P}(Z(X_1) = n) + (1-\alpha) \cdot \mathbb{P}(Z(X_2) = n )
\end{align*}
and for $n=0$
\begin{align*}
\mathbb{P}(W(X)=0) &= \mathbb{P}(X\leq S) \\
  &= \alpha \cdot \mathbb{P}(X_1 \leq S) + (1-\alpha) \cdot \mathbb{P}(X_2 \leq S) \\
  &= \alpha \cdot \mathbb{P}(Z(X_1=0) + (1-\alpha) \cdot \mathbb{P}(Z(X_2)=0).
\end{align*}
Hence
\begin{align*}
\mathbb{E}[W(X)] &= \alpha \, \mathbb{E}[W(X_1)] + (1-\alpha) \, \mathbb{E}[W(X_2)] \\
  &\leq \alpha \cdot a + (1- \alpha) \cdot b.
\end{align*}

Moreover,
\begin{align*}
I(W(X); W(X) + S_2 ) &= H(W(X)+ S_2) - H(W(X)+S_2 | W(X)) \\
  &= H(W(X) +S_2) -H(S_2)\\
  &\overset{(a)}{\geq} \alpha \cdot H(W(X_1)+S_2) + (1-\alpha) \cdot H(W(X_2)+S_2)- H(S_2)\\
  &= \alpha \cdot I(W(X_1);W(X_1)+S_2)  + (1-\alpha) \cdot I(W(X_2);W(X_2)+S_2)\\
  &= \alpha \cdot c(a) + (1-\alpha) \cdot c(b)
\end{align*}
where $(a)$ holds because the distribution of $W(X)+S_2$ is linear in the distribution of $X$ and by concavity of entropy.
Since  $$\mathbb{E}[W(X)] \leq \alpha \cdot a + (1- \alpha) \cdot b$$ we have $$I(W(X);W(X)+S_2) \leq c(  \alpha \cdot a + (1- \alpha) \cdot b)$$ and we deduce the desired result.

In the argument at the beginning of the proof,  if the suprema are not achieved the for an arbitrary $\epsilon > 0$ it suffices to take $X_1$ such that $$I(W(X_1) ; W(X_1) + S_2) > c(a) - \epsilon$$ and $\mathbb{E}[W(X_1)] \leq a$ and similarly for $X_2$ by replacing $a$ by $b$. The same calculation leads to $$I(W(X); W(X) + S_2 ) \geq \alpha \cdot c(a) + (1-\alpha) \cdot c(b) - \epsilon,$$ that is $$ c(  \alpha \cdot a + (1- \alpha) \cdot b) \geq \alpha \cdot c(a) + (1-\alpha)\cdot c(b) - \epsilon$$ and we obtain the desired result by taking $\epsilon\to0$.
\end{proof}

\subsection{Proof of Theorem \ref{suffcond}}
Since
\begin{align*}
 C(\lambda) \leq \lambda \sup_{X:\mathbb{E}(W(X))\leq1/\lambda-1/\mu  \atop W(X) =(X-S_1)^+ } I(W(X),W(X)+S_2)
\end{align*} by 
Corollary \ref{bound} and
since
$$C_F(\lambda) = \lambda \sup_{W\geq0 \atop \mathbb{E}[W]\leq\frac{1}{\lambda}-\frac{1}{\mu}} I(W;W+S),$$ 
if
\begin{align}\label{ineqq}  \sup_{W\geq0 \atop \mathbb{E}[W]\leq\frac{1}{\lambda}-\frac{1}{\mu}} I(W;W+S) > \sup_{X:\mathbb{E}[W(X)]\leq\frac{1}{\lambda}-\frac{1}{\mu} \atop W(X)= (X-S_1)^+}I(W(X);W(X)+S_2)\end{align}
then $$C_F(\lambda) > C(\lambda).$$

For the left-hand side of \eqref{ineqq} we have
$$\sup_{W\geq0 \atop \mathbb{E}[W]\leq\frac{1}{\lambda}-\frac{1}{\mu}} I(W;W+S) = \left[ \sup_{W\geq0 \atop \mathbb{E}[W]\leq\frac{1}{\lambda}-\frac{1}{\mu}} H(W+S) \right] -H(S)$$
since $S$ is independent of $W$.
Similarly, for the right-hand side of \eqref{ineqq} we have
$$\sup_{X:\mathbb{E}(W(X))\leq1/\lambda-1/\mu \atop W(X)= (X-S_1)^+} I(W(X) ; W(X)+S_2) \\ = \left[ \sup_{X:\mathbb{E}[W(X)]\leq\frac{1}{\lambda}-\frac{1}{\mu} \atop W(X)=(X-S_1)^+}H(W(X)+S_2) \right] -H(S)$$
since  $S_2$ is independent of $(X,S_1)$ and since $S_2$ has the same distribution as $S$. The theorem follows.\qedwhite

\subsection{Proof of Example~\ref{binary}}
According to Theorem \ref{suffcond} it is sufficient to show that
\begin{align}\label{lagr}\sup_{W\geq0 \atop \mathbb{E}[W]\leq\frac{1}{\lambda}-\frac{1}{\mu}} H(W+S)>\sup_{X:\mathbb{E}[(X-S_1)^+]\leq\frac{1}{\lambda}-\frac{1}{\mu}}H((X-S_1)^++S_2).
\end{align}
To do this we will show, via Lagrange multipliers, that the optimal distributions of the above two optimization problems are not the same for the two examples.
\subsubsection{Discrete-time model: binary distribution}
 As will be apparent in the proof the values $S$ takes are irrelevant (provided they differ) but analysis gets somewhat simplified when $S\in \{1,2\}$.

Given a non-negative integer valued random variable $W\sim p_W$ let
\begin{align*}
v_n &\defeq \mathbb{P}[W+S=n] 
\end{align*}
which is defined for integers $n\geq 1$ since $W\geq 0$ and $S\in \{1,2\}$. The Lagrangian with respect to the optimization on the left-hand side of \eqref{lagr} is given by 
\begin{multline*}
L_1(p_W, \alpha, \beta) = - \sum_{n=1}^{\infty} v_n \log(v_n)  - \alpha \left( \sum_{k=0}^{\infty} p_W(n) - 1\right) - \beta \left( \sum_{n=0}^{\infty} n p_W(n) - \left( \frac{1}{\lambda}-\frac{1}{\mu} \right) \right).
\end{multline*}

Setting to zero the derivative of $L_1(p_W, \alpha, \beta)$ with respect to $p_W(n)$ we get
\begin{equation}
\label{output with feedback} 
-\frac{1}{2}\log(v_{n+1}) =\frac{1}{2} \log v_n + 1 + \alpha + \beta (n-1) \qquad n\geq 1.
\end{equation}
By solving this recursive equation we obtain for any $n \geq 0$
\begin{align}\label{eqqua1}
v_{1+n} &= e^{-n\beta} \times
\begin{cases}
v_1 \text{ if $n$ is even} \\
v_2 e^\beta = \frac{1}{v_1} e^{-2(\alpha+1)+\beta} \text{ if $n$ is odd} 
\end{cases}
\end{align}
which together with the constraint
\begin{align}
\label{eqqua2}
\sum_{n=1}^{\infty} v_{n} = 1
\end{align}
implies that $\beta>0$ (for otherwise the $\{v_n\}$ do not sum to one). The fact that $\beta>0$ together with the Karush-Kuhn-Tucker condition relative to the constraint $\mathbb{E}[Z]\leq\frac{1}{\lambda}-\frac{1}{\mu}$ implies the equality
\begin{align}\label{eqqua3}\sum_{n=0}^{\infty} (1+n) v_{1+n} = \frac{1}{\lambda}.
\end{align}
Equations \eqref{eqqua1},\eqref{eqqua2} and Equations \eqref{eqqua1},\eqref{eqqua3} imply
\begin{align}\label{eqqua4}
 v_1 + v_2 = 1- e^{-2\beta}
\end{align}
and
\begin{align}\label{eqqua5}
v_2+(1+\lambda)v_1 = 2 \lambda,
\end{align}
respectively. 

Similarly, we define the Lagrangian with respect to the right-hand side of~\eqref{lagr}
\begin{multline*}
L_2(p_X,\gamma,\delta) =  - \sum_{n=1}^{\infty} w_n \log w_n   - \gamma \left( \sum_{n=0}^{\infty} p_X(n) - 1\right)  \\ -\delta \left( \left[ \sum_{k=0}^{\infty} n ( \frac{1}{2} \, p_X(n+1) + \frac{1}{2} \, p_X(n+2)) \right] - \left[ \frac{1}{\lambda}-\frac{1}{\mu} \right]\right)
\end{multline*}
where 
\begin{align*}
{w}_n &\defeq \mathbb{P}[(X-S_1)^++S_2=n] \quad n\geq 1.
\end{align*}
Setting to zero the derivative of $L_2(p_X,\gamma,\delta)$ we get 
\begin{align}\label{eqfo1}
 - \frac{\log w_1+\log w_2}{2}  =1 +\gamma
\end{align}
and
\begin{align}
-\log  w_{n+1}-2 \log w_n - \log w_{n-1}  =4+ 4\gamma + 2\delta ( 2n-3)   \label{output without feedback}
\end{align}
for $ n \geq 2 $.
The sequence $\{-\log w_n\}_{n\geq 1}$ is therefore the solution of an order two non-homogeneous recursive equation of characteristic polynom $X^2+2X+1 = (X+1)^2$.
Hence there exist two constants $c_1$ and $c_2$ and a particular solution $\{p_n\}_{n\geq 1}$ of the non-homogeneous equation such that 
\begin{align}
-\log w_n = (-1)^n(c_1+c_2 n) + p_n. \label{wn}
\end{align}
It can be verified that a particular solution is given by
\begin{align*}
p_n &= 1 +  \gamma - \frac32 + \delta n 
\end{align*}
or, equivalently,
\begin{align}\label{eqfo2}
p_n & = - \frac{\log w_1 + \log w_2 }{2} - \frac32 +\delta n
\end{align}
by \eqref{eqfo1}.  From \eqref{wn} with $n=1,2$ and \eqref{eqfo2} we get
\begin{align}\label{c1}
c_1 = \frac{\log w_1 - \log w_2}{2}- \frac92 +4\delta
\end{align}
\begin{align}\label{c2}
c_2 = 3(1-\delta ).
\end{align}

Now let us assume that  $$w_n = v_n$$ for all $n\geq 1$.
By injecting (\ref{output with feedback}) into (\ref{output without feedback}) we obtain that $$\alpha = \gamma \quad \text{ and } \quad \beta = \delta.$$
Combining \eqref{eqqua1} and \eqref{wn} we get $c_2=0$ and, equivalently, $ \delta=1$ by \eqref{c2}. 
Hence, from \eqref{wn} and \eqref{c1} we have
\begin{align*}
  -\log w_1 = - \log w_2 -1.
\end{align*}
This together with \eqref{eqqua4} yields
$$w_1 = \frac{1 - e^{-2}}{1+e}$$
and therefore from \eqref{eqqua5}
$$\lambda = \frac{(1+e)(1-e^{-2})}{1+e^{-2}+2e}.$$
Hence,  if $\lambda \neq \frac{(1+e)(1-e^{-2})}{1+e^{-2}+2e}$ the distributions $\{v_n\}$ and $\{w_n\}$ are not equal. This in turn implies that
$$\sup_{X:\mathbb{E}[(X-S_1)^+]\leq\frac{1}{\lambda}-\frac{1}{\mu}}H((X-S_1)^++S_2) < \sup_{Z\geq0 \atop \mathbb{E}[W]\leq\frac{1}{\lambda}-\frac{1}{\mu}} H(W+S)$$
 and it follows that for any $\lambda \neq \frac{(1+e)(1-e^{-2})}{1+e^{-2}+2e}$
$$C_{WF}(\lambda) < C_F(\lambda).$$

It can be verified that for $\lambda = \frac{(1+e)(1-e^{-2})}{1+e^{-2}+2e}$ we have that $\frac{\partial C_F(\lambda)}{\partial \lambda}$ or, equivalently, 
$$\frac{\partial H(W+S)}{\partial \lambda}=\frac{-\partial{ \sum_{n=1}^{\infty} v_n \log(v_n)}}{\partial \lambda}$$
differs from zero.
Hence $C_F$ is not achieved for $\lambda = \frac{(1+e)(1-e^{-2})}{1+e^{-2}+2e}$ and therefore
$$C_{WF}<C_F$$ which in turns implies $$C<C_F.$$
\qedwhite

\subsubsection{Continuous-time model: uniform distribution}
We compute the distribution of $W+S$ that achieves the left-hand side of \eqref{lagr}.

Since $S$ is uniform over $[0,1]$, $W+S$ has density
$$f_{W+S}(t) =F_W(t)-F_W(t-1)$$
where $F_W$ denotes the cumulative density function of $W$.
The Lagrangian corresponding to the left-hand side of  \eqref{lagr} is
$$L_1(\mathbb{P}_W, \alpha, \beta)=- \int_t f_{W+S}(t) \log f_{W+S}(t) dt - \alpha \int_{\mathbb{R}} d\mathbb{P}_W - \beta \int_{t \in \mathbb{R}}t d\mathbb{P}_W, $$ and its G\^ateaux derivative with respect to a measure $\nu$ (see, {\it{e.g.}}, \cite[Chapter 7.2]{gateaux}) is given by
\begin{align} \label{gat1}
\delta L_1(\mathbb{P}_W,\alpha, \beta ; \nu) = \int_t A_\nu(t) (1+ \log A_{\mathbb{P}_W}(t)) dt - \alpha \int_{\mathbb{R}}d\nu - \beta \int_{t \in \mathbb{R}} t d\nu
\end{align}
where $$A_\nu: t \mapsto \nu(]t-1;t])$$
for any measure $\nu$.
The cone of the distributions $W$ such that $\mathbb{E}[W] \leq \frac{1}{\lambda}-\frac{1}{\mu}$ contains interior points because $\lambda < \mu$. Hence we can apply the generalized Kuhn-Tucker Theorem \cite[Chapter 9.4]{gateaux} which says that for any $W$ that maximizes the left-hand side of \eqref{lagr} expression \eqref{gat1} must be equal to zero for any measure $\nu$. In particular, if we let $\nu$ be a Dirac $\delta_{t_0}$ at some $t_0 \in \mathbb{R}^+$ we get
$$\delta L_1(\mathbb{P}_W,\alpha, \beta ; \delta_{t_0}) = - \int_{t=t_0}^{t_0 + 1} (1+\log  A_{\mathbb{P}_W}(t)) dt - \alpha - \beta t_0=0.$$ 
Differentiating with respect to $t_0$ we get
$$-\log  A_{\mathbb{P}_W}(t_0+1) +\log  A_{\mathbb{P}_W}(t_0) = \beta.$$
We note that $\beta \ne 0$ because otherwise $\mathbb{P}_W$ does not sum to one. Using the fact that $S$ is uniform over $[0,1]$ we get $A_{\mathbb{P}_W} = f_{W+S}$ and therefore we deduce that $W+S$ should satisfy
\begin{align} \label{ffb}
f_{W+S}(t_0 +1) = \frac{f_{W+S}(t_0)}{\beta} \qquad t_0 >0. 
\end{align}

We now compute the distribution of $X$ that achieves the right-hand side of \eqref{lagr}.
Setting the G\^ateaux derivative of the Lagrangian to zero gives
\begin{align} \label{gat2}
\delta L_2(\mathbb{P}_X,\gamma, \delta ; \nu)=\int_t B_\nu(t) (1+ \log B_{\mathbb{P}_X}(t)) dt - \gamma \int_{\mathbb{R}}d\nu - \delta \int_{t \in \mathbb{R}} t d\nu =0
\end{align}
where $$B_\nu : t \mapsto \int_{x=0}^1 \nu( ]t-x ; t-x+1])dx $$
for any measure $\nu$.
As previously, by letting $\nu$ be a Dirac $\delta_{t_0}$ and using the fact that $B_{\mathbb{P}_X} = f_{(X-S_1)^+ +S_2}$ we get 
\begin{align} \label{fwfb}
f_{(X-S_1)^++S_2}(t_0 +1) = \frac{f_{(X-S_1)^+ +S_2}(t_0-1)}{\delta} \quad t_0>0.
\end{align}
Now suppose that
\begin{align}\label{hypothese}
\sup_{W \geq 0 \atop \mathbb{E}[W] \leq \frac{1}{\lambda}-\frac{1}{\mu}} H(W+S)  = \sup_{X : \mathbb{E}[(X-S_1)^+] \leq \frac{1}{\lambda}-\frac{1}{\mu}} H((X-S_1)^++S_2). 
\end{align}

This implies that that there exists $W$ and $X$ that satisfy the above constraints and such that $$f_{(X-S_1)^++S_2}=f_{W+S}$$ and from \eqref{ffb} and \eqref{fwfb} we get
\begin{align} \label{deqb2}
 \delta = \beta ^2.
\end{align}
Moreover, let us define $$\nu_1 : ]a;b] \mapsto \int_{t=a}^b \nu_2(]t;t+1] dt$$
for any let $\nu_2$ such that $\int_\mathbb{R} d_{\nu_2}=1$. It then follows that for any such $\nu_2$ $$\delta L_1(\mathbb{P}_Z,\alpha, \beta ; \nu_1) - \delta L_2(\mathbb{P}_X,\gamma, \delta ; \nu_2)=\gamma -\alpha +(\delta-\beta) \int_{t=0}^{\infty} t \nu_2(]t;t+1]) dt=0.$$ Hence,
$\beta = \delta \quad (\text{and} \;\alpha = \gamma).$
This together with \eqref{deqb2} and the fact that $\beta \neq 0$ (see above) implies that $$\beta=1.$$
This by \eqref{ffb} implies that $f_{W+S}$ does not sum to $1$ and is therefore not a density. Because of this contradiction we deduce that \eqref{hypothese} does not hold, which concludes the proof.\qedwhite

\subsection{Proof of Theorem~\ref{th:suffcondequal}}

The following lemma is essentially a consequence of the well-known fact (see \cite{maxentropy}) that the geometric random variable maximizes the entropy among all discrete random variables with the same mean:
\begin{lemma} \label{maxh} 
Suppose $X^n$ satisfies
$$\frac{\sum_{k=1}^n \mathbb{E}[X_k]}{n} \leq \frac{1}{\lambda}.$$
Then
\begin{align*} 
\frac{H(X^n)}{n} \leq H(g_{\lambda}) = - \frac{(1-\lambda)\log(1-\lambda) + \lambda \log(\lambda)}{\lambda}
\end{align*}
where $g_{\lambda}$ denotes the geometric distribution with mean $1/{\lambda}$. Moreover, the above inequality is tight if and only if $X^n$ is i.i.d. geometrically distributed.
\end{lemma}

We first show that 
$$C(\lambda)= \lambda (H(g_\lambda) - H(S))\qquad \lambda < \mu. $$
We start with the achievability part. A general capacity result \cite{general_formula_for_capacity} gives us  
\begin{equation} 
C(\lambda)= \lambda \sup_{\bf{A}} \underline{I}(\bf{A} ; \bf{D}) \label{gencap}
\end{equation}
where ${\bf{D}}=\{D^n\}_{n \geq1}$ is the interdeparture process of the queue when interarrival process is ${\bf{A}}=\{A^n\}_{n \geq1}$. In the above expression $\underline{I}$ is the inf-information rate between {\bf{A}} and {\bf{D}}, that is the $\liminf$ in probability of the sequence of normalized information densities $\frac1n i(A^n ;D^n)$ where
$$i(A^n;D^n) \defeq \log \frac{\mathbb{P}(D^n| A^n)}{\mathbb{P}(D^n)}.$$
We similarly consider the  inf- and sup-entropy rate $\underline{H}$ and $\overline{H}$, respectively (see \cite{han1993approximation} for the precise definitions of $\underline{I}$, $\underline{H}$, and $\overline{H}$). From \cite[Theorem~$8$]{general_formula_for_capacity}
\begin{align} \label{ineqi}
\underline{H}({\bf D}) - \underline{H}({\bf D}|{\bf A})  \leq \underline{I}({\bf A};{\bf D}) \leq \overline{H}({\bf D}) - \underline{H}({\bf D}|{\bf A}).
\end{align}
Now suppose that $\bf{A}$ is i.i.d. according to a  mean ${1}/{\lambda}$ geometric distribution with $\lambda<\mu$. This implies that the queue is stable and therefore 
$$\frac1n \sum_{k=1}^n\mathbb{E}[D_k] \overset{n\to \infty}{\to} \frac{1}{\lambda}.$$
Hence by Lemma~\ref{maxh}
 $$ \limsup_{n\to \infty}\frac{1}{n}H(D^n)\leq H(g_\lambda).$$
Since by assumption 
 $$H(g_\lambda) \leq \liminf_{n\to \infty}\frac{1}{n}H(D^n),$$
 it follows that
  $$H(g_\lambda) = \lim_{n\to \infty}\frac{1}{n}H(D^n)=\overline{H}({\bf D})=\underline{H}({\bf D})$$
and {\bf D} is i.i.d. geometric.
Hence by \eqref{ineqi}
$$\underline{I}({\bf{A}}; {\bf{D}})= H(g_\lambda) - H(S), $$
and from \eqref{gencap} we conclude the direct part of the theorem
\begin{align}\label{direct}
C(\lambda)\geq \lambda (H(g_\lambda) - H(S)). 
\end{align}

For the converse part, a standard application of Fano's inequality gives 
\begin{equation} \label{shannon}
C(\lambda) \leq \lambda \liminf_{n \to \infty} \sup_{A^n} \frac1n I(A^n;D^n)
\end{equation}
where
\begin{align*}
I(A^n ; D^n) &= H(D^n) - nH(S).
\end{align*}
From the Theorem's assumption  it follows that $$\frac{1}{n}\sum_{i=1}^n{\mathbb{E}} D_i\leq \frac{1}{\lambda}.$$ Hence, from Lemma~\ref{maxh} and the fact that $H(g_\lambda)$ is non-increasing in $\lambda$ we get
$$\limsup_{n \to \infty} \frac{H(D^n)}{n} \leq H(g_{\lambda})$$
for any interarrival times $A^n$.
This together with \eqref{shannon} implies  the converse 
\begin{align} \label{converse}
C(\lambda) \leq \lambda (H(g_\lambda) - H(S)).
\end{align}
From \eqref{direct} we then conclude that
$$C(\lambda) = \lambda (H(g_\lambda) - H(S)).
$$ 

To conclude the proof of the theorem it suffices to show that 
$$ C_F(\lambda) \leq \lambda (H(g_\lambda) - H(S)).$$
From \cite{bits_through_queues} we have
 $$ C_F = \lambda \sup_{X \geq 0   \atop   \mathbb{E}[X] \leq \frac{1}{\lambda} - \frac{1}{\mu}} I(X ; X + S).$$
Now, 
\begin{align*}
I(X ; X+S) &= H(X+S) - H(X+S | X)\\
  &= H(X+S) - H(S)\\
  &\leq H(g_{\frac{1}{\mathbb{E}[X+S]}}) - H(S)\\
  &\leq H(g_\lambda) -H(S)
\end{align*}
where the second inequality holds since $H(g_\lambda)$ is non-decreasing in $1/\lambda$. It then follows that for any $\lambda $ we have $C(\lambda)=C_F(\lambda)$ and therfore $$C = \sup_{\lambda < \mu}C(\lambda)    = \sup_{\lambda < \mu}C_F(\lambda)=C_F.$$
\qedwhite


%
\section{Concluding remarks}\label{concrem}
In this paper we introduced a notion of weak feedback and provided a simple upper bound to the weak feedback capacity which, in certain cases, allows to distinguish the capacities with and without feedback. An interesting open question is whether this bound is actually equal to the weak feedback capacity. Insights into this question would shed light on necessary and sufficient conditions under which feedback increases capacity of timing channels.

\section*{Acknowledgement}
The authors would like to thank Venkat Anantharam for stimulating discussions which prompted the current investigation.

\section*{Appendix}

We provide a direct proof of Corollary~\ref{bound}. 
\begin{align*}
 C(\lambda) \leq \lambda \sup_{X:\mathbb{E}(W(X))\leq1/\lambda-1/\mu  \atop W(X) = (X-S_1)^+ } I(W(X),W(X)+S_2)
\end{align*}

\begin{proof}
Denote by $U$ the transmitted message and by $V$ the decoded message. Assuming that $U$ is equiprobable, Fano's inequality and the data processing inequality \cite{Cover_Thomas} we have that that every $(n,M,n/\lambda,\epsilon)$-code satisfies
\begin{align}\label{fanoss}
\log M &\leq\frac{1}{1-\epsilon}[I(U;V) + 1] \\
   & \leq\frac{1}{1-\epsilon}[I(A^n;D^n)+ 1].
\end{align}
In \cite{bits_through_queues} it is shown that 
\begin{align} I(A^n;D^n) \leq \sum_{i=1}^n I(W_i;W_i+S_i)\label{av}
\end{align}
where 
\begin{align}\label{vi}W_i =(a_i-d_{i-1})^+
\end{align}
is the waiting time of the queue between the departure of the $(i-1)$\/th packet and the beginning of the service of the $i$\/th packet.

Now suppose that $W_{i-1} \neq 0.$ Then $$d_{i-1}=a_{i-1}+S_{i-1}$$ 
and therefore
\begin{align*}
W_i &= a_i - d_{i-1}\\
  &=a_i - a_{i-1}-S_{i-1}\\
  &=A_i-S_{i-1}.
\end{align*}
It follows that if $V_{i-1} \neq 0$ then $W_i = (A_i-S_{i-1})^+$.

Similarly, if $W_{i-1} = 0$ but $W_{i-2} \neq 0$ 
\begin{align}\label{di}
d_{i-1} = a_{i-2} + S_{i-2} + S_{i-1}
\end{align}
which implies that
$$W_i = ((A_{i-1}+A_i-S_{i-2})-S_{i-1})^+.$$

Now let $$k\defeq\sup\{j<i:W_j \neq 0\}$$ with the convention that $k=1$ if $\{j<i:W_j \neq 0\}=\emptyset$.
Then from \eqref{vi} and by iterating \eqref{di} we get
\begin{align*}
W_i &= \left( \sum_{j=k+1}^i A_j - \sum_{j=k}^{i-1} S_j \right)^+\\
  &=  \left( \left[ \sum_{j=k+1}^i A_j - \sum_{j=k}^{i-2} S_j \right] - S_{i-1} \right)^+.
\end{align*}
Now define $$X_i\defeq  \sum_{j=k+1}^i A_j - \sum_{j=k}^{i-2} S_j.$$ From \eqref{av} we get
$$I(A^n;D^n) \leq  \sum_{i=1}^n I((X_i-S_{i-1})^+;(X_i-S_{i-1})^++S_i).$$
As in the proof of Theorem~\ref{wfc} consider the function
$$c(a) \defeq \sup_{X:\mathbb{E}[W(X)] \leq a} I((X-S_1)^+ ; (X-S_1)^++S_2)$$
where $W(X)=  (X-S_1)^+$, where $X$ is a random variable independent of $(S_1,S_2)$. Recalling that $c$ is non-decreasing and concave we get
$$I(U;V) \leq \sum_{i=1}^n c(\mathbb{E}[W(X_i)])$$ and therefore
\begin{align}\label{mutua}
\frac1n I(U;V) &\leq \frac1n \sum_{i=1}^n c(\mathbb{E}[W_i])\\
  &\overset{(a)}{\leq} c \left( \sum_{i=1}^n \frac1n \mathbb{E}[W_i] \right)\\
  &\overset{(b)}{\leq} c \left( \frac1\lambda - \frac1\mu \right)
\end{align}
where $(a)$ follows from the concavity of $c$, 
where $(b)$ is a consequence of the fact that we impose an output rate no greater than $\lambda$ and that $c$ is a non-decreasing.
From \eqref{fanoss} and \eqref{mutua} we get that the rate of any $(n,M,n/\lambda,\epsilon)$-code satisfies
\begin{align*}
\lambda \frac{\log M}{n} \leq \frac{\lambda}{1-\epsilon}\left(c\left(\frac{1}{\lambda}-\frac{1}{\mu} \right)+\frac{1}{n} \right)
\end{align*}
what implies the desired result.
\end{proof}

\bibliographystyle{plain}
\bibliography{biblio}

\end{document}